\documentclass[final,3p,times,twocolumn]{elsarticle}

\usepackage{hyperref}
\usepackage{graphicx}
\usepackage{amsthm, amsfonts, amsmath, amssymb}
\usepackage[ruled, linesnumbered, vlined]{algorithm2e}
\usepackage[colorinlistoftodos]{todonotes}
\usepackage{enumitem}
\usepackage{cleveref}

\newcommand{\leo}[1]{#1}
\newcommand{\leonie}[1]{#1}

\newcommand{\cT}{\mathcal{T}}
\newcommand{\cD}{\mathcal{D}}
\newcommand{\cP}{\mathcal{P}}

\newcommand{\tb}{\overline{\mathcal{T}}}
\newcommand{\getss}{=}

\newtheorem{theorem}{Theorem}
\newtheorem{lemma}{Lemma}

\theoremstyle{definition}
\newtheorem{definition}{Definition}

\journal{a journal}

\begin{document}

\begin{frontmatter}



\title{An algorithm for reconstructing level-2 phylogenetic networks from trinets}


\author[inst1]{Leo van Iersel\fnref{fn1}}
\ead{L.J.J.vanIersel@tudelft.nl}
\author[inst1]{Sjors Kole}
\author[inst2]{Vincent Moulton}
\ead{V.Moulton@uea.ac.uk}
\author[inst1]{Leonie Nipius}

\affiliation[inst1]{organization={Delft Institute of Applied Mathematics, Delft
    University of Technology},
            addressline={Mekelweg 4}, 
            city={Delft},
            postcode={2628 CD}, 
            country={The Netherlands}}
            
            \affiliation[inst2]{organization={School of Computing Sciences, University of East Anglia},
            addressline={NR4 7TJ}, 
            city={Norwich},
            country={United Kingdom}}

\fntext[fn1]{Research funded in part by the Netherlands Organization for Scientific Research (NWO) Vidi grant 639.072.602.}

\begin{abstract}
Evolutionary histories for species that cross with one another or exchange genetic material can be represented by leaf-labelled, directed graphs called \emph{phylogenetic networks}. A major challenge in the burgeoning area of phylogenetic networks is to develop algorithms for building such networks by amalgamating small networks into a single large network. The \emph{level} of a phylogenetic network is a measure of its deviation from being a tree; the higher the level of network, the less treelike it becomes. Various algorithms have been developed for building level-1 networks from small networks. However, level-1 networks may not be able to capture the complexity of some data sets. In this paper, we present a polynomial-time algorithm for constructing a rooted binary level-2 phylogenetic network from a collection of 3-leaf networks or \emph{trinets}. Moreover, we prove that the algorithm will correctly reconstruct such a network if it is given all of the trinets in the network as input. The algorithm runs in time $O(t\cdot n+n^4)$ with~$t$ the number of input trinets and~$n$ the number of leaves. We also show that there is a fundamental obstruction to constructing level-3 networks from trinets, and so new approaches will need to be developed for constructing level-3 and higher level-networks.
\end{abstract}




\end{frontmatter}


\section{Introduction}

Phylogenetic networks are a generalization of phylogenetic trees that are commonly used to represent the evolutionary histories of species that cross with one another or exchange genetic material, such as plants and viruses. There are several classes of phylogenetic networks and various ways have been devised to build them – see e.g. \cite{Elworth2019,steel2016phylogeny} for recent surveys. Mathematically speaking, a {\em phylogenetic network} on a set of species $X$ is basically a directed acyclic graph, with a single source or {\em root}, such that every sink or {\em leaf} has indegree 1 and the set of leaves is equal to $X$. In this paper, we shall only consider {\em recoverable, binary} networks (or {\em networks} for short), that is, phylogenetic networks that satisfy a certain condition on the ancestors of $X$ and in which the root has outdegree 2 and all other, non-leaf, vertices have degree 3 (see Figure~\ref{fig:counter_example} for some examples). Precise definitions are given in Section~\ref{sec:prelim}.

\begin{figure*}
    \centering
    \includegraphics{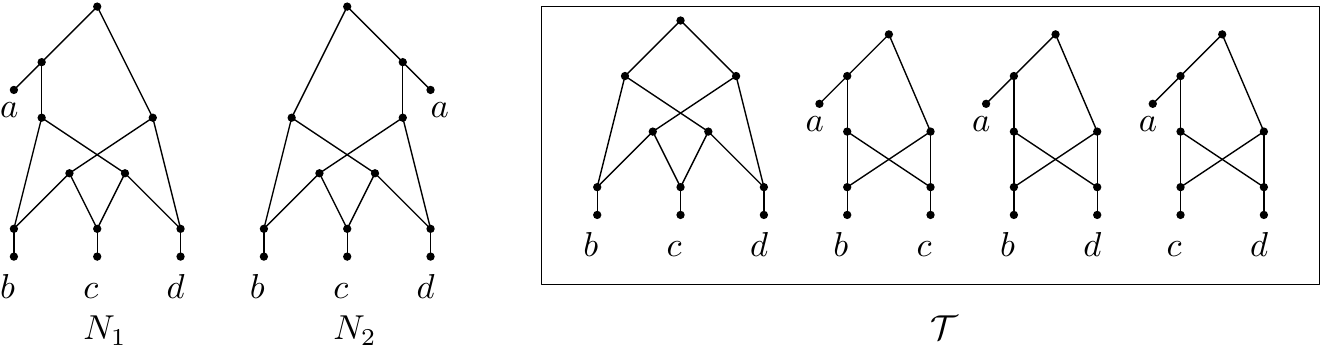}
    \caption{Left: Two distinct level-3 networks $N_1$ and $N_2$ on the set $X=\{a,b,c,d\}$. Right: The set of 
    trinets~$\cT$ that is contained in both $N_1$ and $N_2$.\label{fig:counter_example}}
\end{figure*}


Recently, there has been growing interest in the problem of building a network with leaf-set $X$ from a collection of networks each of which having leaf-set equal to some subset of $X$ in such a way that the input networks are each contained in the final network. Early work on this so-called {\em supernetwork problem} focused on building up networks from {\em phylogenetic trees}, that is, phylogenetic networks whose underlying graph is a tree. Several results have been presented for this problem, including algorithms for constructing networks from triplets, \leonie{which are} 3-leaved phylogenetic trees, (e.g. \cite{huber2010practical}) and from collections of phylogenetic trees all on leaf-set $X$ (e.g. \cite{willson2010regular}) – for a recent summary of these approaches see \cite{semple2021trinets}. However, an important issue with this strategy is that phylogenetic trees do not necessarily {\em encode} phylogenetic networks, i.e., there are examples of distinct (non-isomorphic) networks that contain the same set of phylogenetic trees (see e.g. \cite{gambette2012encodings}), making it impossible to uniquely reconstruct such networks from their trees. 

Motivated by this issue, in \cite{huber2013encoding} it was proposed to build networks from collections of 3-leaved networks, or {\em trinets}. In that paper, the authors focused on building level-1 networks\footnote{In fact they considered the somewhat more general class of 1-nested networks.} where, in general, {\em level-$k$ networks} are networks that can be converted into a tree by deleting at most $k$ arcs from each biconnected component. In particular, they showed that level-1 networks are encoded by the trinets that they contain, and gave an algorithm for constructing a level-1 network on $X$ from its trinets that is polynomial in $|X|$ (see also \cite{trilonet} for a more general algorithm). In \cite{level2trinets} the encoding result was extended to the more general class of level-2 networks, and also to the distinct and quite broad class of so-called {\em tree-child} networks. Recently, in \cite{semple2021trinets} it was also shown that {\em orchard} networks, which generalise tree-child networks, are encoded by their trinets, and an algorithm was given for constructing an orchard network from its trinets that is polynomial in the size of the vertex set of the network (whose size is not necessarily polynomial in $|X|$).

Intriguingly, in \cite{huber2015much} it was shown that, as with trees, trinets do {\em not} encode networks in general. Indeed, in \cite[p.28]{semple2021trinets} it was shown that even level-4 networks are not encoded by their trinets and, since level-2 networks are encoded by their trinets (see above), it was asked whether or not level-3 networks are encoded by their trinets (see also \cite{dagstuhl}). In the first result of this paper we answer this question – in particular, the two networks $N_1$ and $N_2$ in Figure~\ref{fig:counter_example} are level-3 and are easily seen to be distinct and to contain the same set of trinets \leonie{(see~\cite{leonie})}. Hence, level-$k$ networks are encoded by their trinets only if $k \le 2$. As the algorithm in \cite{huber2013encoding} can be used to uniquely reconstruct a level-1 network from its trinets, this leaves open the question of finding a polynomial algorithm for building a level-2 network from its trinets, which is the purpose of the rest of this paper. In particular, we shall present an algorithm which constructs a level-2 network on $X$ from {\em any} set of trinets $\cT$ whose leaf-set union is $X$ that runs in $O(|\cT||X|+|X|^4)$ time (Algorithm~\ref{alg:alg}) and that is guaranteed to reconstruct a level-2 network from its set of trinets (Theorem~\ref{thm:induced}). We now proceed by presenting some preliminaries, after which we shall describe our level-2 algorithm. We will conclude with a brief discussion of our results.


\section{Preliminaries}\label{sec:prelim}

We refer the reader to \cite[Chapter 10]{steel2016phylogeny} for more information on the terminology and basic
results on phylogenetic networks that we summarise in this section.

\begin{definition}
Let~$X$ be some finite set (corresponding to a set of species, say). 
A \emph{binary phylogenetic network (on~$X$)} is a directed acyclic graph with the following types of vertices:
a single \emph{root} with indegree~0 and outdegree~2;
\emph{tree-vertices} with indegree~1 and outdegree~2;
\emph{reticulations} with indegree~2 and outdegree~1; and
\emph{leaves} with indegree~1 and outdegree~0,
where the leaves are in one-to-one correspondence with the elements of~$X$.
\end{definition}

Let~$N$ be a binary phylogenetic network on~$X$, and suppose that $u,v$  are two vertices in the vertex set $V(N)$ of $N$. If there is a directed path from~$u$ to~$v$ (including the case that~$u=v$), then we say that~$u$ is an \emph{ancestor} of~$v$ and that~$v$ is a \emph{descendant} of~$u$. When~$(u,v)$ is an arc, we say that~$u$ is a \emph{parent} of~$v$ and that~$v$ is a \emph{child} of~$u$. We say that~$(u,v)$ is a \emph{cut-arc} if deleting~$(u,v)$ disconnects~$N$. A set~$A\subseteq X$ is called a \emph{cut-arc set} in~$N$ if~$A=X$ or~$A$ is the set of descendant leaves of~$v$ for some cut-arc~$(u,v)$. A cut-arc set~$A$ is \emph{minimal} if~$|A|>1$ and there is no cut-arc set~$B$ with~$|B|>1$ and~$B \subsetneq A$. A network is \emph{simple} if it has no minimal cut-arc set.

Now, suppose $A \subseteq X$. A \emph{lowest stable ancestor (LSA)} of~$A$ in~$N$ is a vertex~$v$ such that, for all~$a\in A$, all paths from the root to~$a$ contain~$v$, and such that there is no descendant~$u$ of~$v$ with $u\neq v$ that satisfies this property. It is not difficult to see that the lowest stable ancestor is always unique for any~$A\subseteq X$ \cite[p.263]{steel2016phylogeny}. We say that $N$ is \emph{recoverable} if $LSA(X)$ is the root of $X$. In this paper, for simplicity, we shall call a recoverable, binary phylogenetic network on $X$ a \emph{network}. \leo{Only in statements of theorems we will mention these restrictions explicitly.}

A \emph{biconnected component} of a network is a maximal subgraph not containing any cut-arcs. A network is \emph{level-$k$} if each biconnected component contains at most~$k$ reticulations. A level-$k$ network is \emph{strictly level-$k$} if it is not level-$k'$ for any $k'<k$. This paper will mainly focus on level-2 networks; see Figure~\ref{fig:cutarcs} for an example.

A network on~$A$ is a \emph{trinet} if~$|A|=3$ and a \emph{binet} if~$|A|=2$. If~$T$ is a trinet or binet on~$A$ then we also use~$L(T)$ to denote the set~$A$. Furthermore, for a set of trinets and/or binets~$\cT$, we define $L(\cT)=\cup_{T\in\cT}L(T)$. We will now define the restriction of a network to a subset of~$X$, which will be used to define the set of trinets contained in a network.

\begin{definition}
Let~$N$ be a network on~$X$ and~$A\subseteq X$. The \emph{restriction} of~$N$ to~$A$, denoted~$N|A$, is the network on~$A$ obtained from~$N$ by \leonie{deleting all vertices that are not on a path from $LSA(A)$ to an element of~$A$ and subsequently replacing parallel arcs by single arcs and suppresssing indegree-1 outdegree-1 vertices, until neither of these operations is applicable.}
\end{definition}

The set of trinets~$\cT(N)$ of a network~$N$ on~$X$ is defined as $\{N|A \mid A\subseteq X, |A|=3\}$. The set of binets and trinets~$\tb(N)$ of a network~$N$ on~$X$ is defined as $\{N|A \mid A\subseteq X, 2\leq |A|\leq 3\}$. Observe that $\tb(N)$ can be obtained from $\cT(N)$.

We say that two networks~$N,N'$ on~$X$ are \emph{equal} and write~$N=N'$ if there is an isomorphism~$f:V(N)\rightarrow V(N')$ such that, for all~$x\in X$,~$f(x)$ has the same label as~$x$.

The following theorem forms the basis for our new level-2 algorithm.

\begin{theorem}[\cite{level2trinets}]
Let~$N$ be a recoverable, binary level-2 network on~$X$ with~$|X|\geq 3$. Then there exists no recoverable network~$N'\neq N$ with $\cT(N)=\cT(N')$.
\end{theorem}


\subsection{Generators}

Our algorithm will make heavy use of the underlying structure of biconnected components, which is called a ``generator'' and defined as follows.

\begin{definition}
Let~$N$ be a simple network. The \emph{underlying generator} of~$N$ is the directed multigraph~$G$ obtained from~$N$ by deleting all leaves and suppressing all indegree-1 outdegree-1 vertices. The arcs and indegree-2 outdegree-0 vertices of~$G$ are called \emph{sides}. The arcs are also called \emph{arc sides} and the indegree-2 outdegree-0 vertices also \emph{reticulation sides}. We say that leaf~$x$ is \emph{on side}~$S$ (or that side~$S$ \emph{contains}~$x$) if either
\begin{itemize}
    \item $S$ is a reticulation side of~$G$ and the parent of~$x$ in~$N$, or
    \item $S$ is an arc side of~$G$ obtained by suppressing indegree-1 outdegree-1 vertices of a path~$P$ in~$N$ and the parent of~$x$ lies on path~$P$.
\end{itemize}
\end{definition}

See Figure~\ref{fig:gen} for all underlying generators of simple level-1 and level-2 networks.

\begin{figure}[h]
    \centering
    \includegraphics{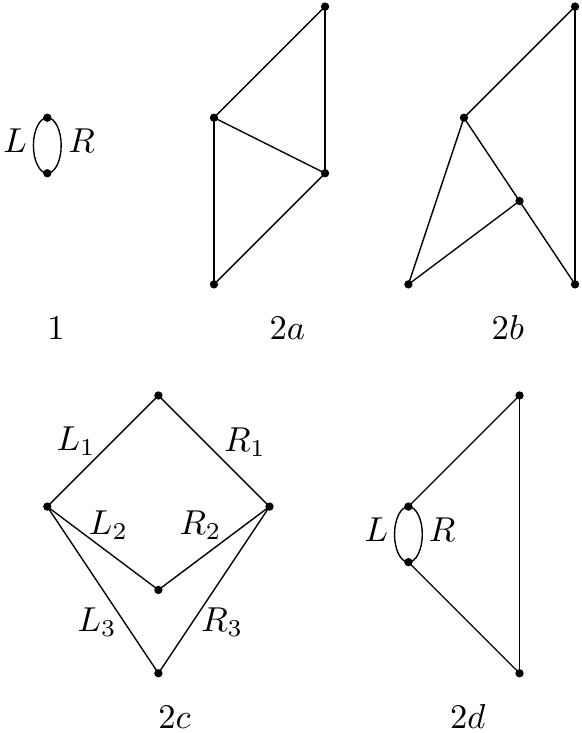}
    \caption{The only underlying generator of a simple level-1 network and the four underlying generators of simple level-2 networks~\cite{level2trinets}. Generator $2c$ has three sets of symmetric arc sides $\{L_1,R_1\},\{L_2,R_2\},\{L_3,R_3\}$ while generators~$1$ and~$2d$ have one set of symmetric arc sides~$\{L,R\}$. Generator~$2c$ is the only level-2 generator with symmetric reticulation sides.\label{fig:gen}}
\end{figure}


To \emph{attach} leaf~$x$ to a reticulation side~$S$ means adding~$x$ with an arc from~$S$ to~$x$. To \emph{attach} a list~$(x_1,\ldots ,x_l)$ of leaves to an arc side~$S$ means subdividing~$S$ to a path with~$l$ internal vertices~$p_1,\ldots ,p_l$ and adding leaves $x_1,\ldots ,x_l$ with arcs $(p_1,x_1),\ldots ,(p_l,x_l)$.

A trinet~$T\in\cT(N)$ is called a \emph{crucial} trinet of a simple network~$N$ if it contains a leaf on each arc side of the underlying generator~$G$ of~$N$ and, for each pair of parallel arcs in~$G$, a leaf on at least one of these two sides. Crucial trinets are of special interest because they have the same underlying generator as the network~$N$.

Two reticulation sides~$u,v$ of a generator~$G=(V,A)$ are \emph{symmetric} if there exists an automorphism $f:V\to V$ of~$G$ with~$f(u)=v$. The equivalence classes under this notion of symmetry are called \emph{sets of symmetric reticulation sides}.

Two arc sides~$(u,v)$, $(u',v')$ of a generator~$G=(V,A)$ are \emph{symmetric} if there exists an automorphism $f:V\to V$ of~$G$ with~$f(r)=r$ for each reticulation side~$r$ and such that~$u'=f(u)$ and~$v'=f(v)$. The equivalence classes under this notion of symmetry are called \emph{sets of symmetric arc sides}. For an example, see Figure~\ref{fig:gen}. The idea behind this definition is that the reticulation sides of~$G$ are parents of leaves in~$N$. In our algorithm, we will make heavy use of crucial trinets, which contain those leaves. Since they are labelled, we can distinguish them. 

\section{Algorithm}


\subsection{Outline}\label{sec:combine}

We work with multisets of trinets and binets because these may arise when collapsing or restricting trinet sets. Hence, let~$\cT$ be a multiset of binets and trinets. The high-level idea of the algorithm is to first find a minimal cut-arc set~$A$. Then we construct~$\cT^*$ by collapsing~$A$ to a single leaf~$a^*$ and find a network~$N^*$ for~$\cT^*$ recursively. The next step is to construct~$\cT'$ from~$\cT$ by restricting to the taxa in~$A$ and to find a simple network~$N'$ for~$\cT'$. Finally, we construct~$N$ from~$N^*$ and~$N'$ by replacing~$a^*$ by~$N'$. The pseudo code is in Algorithm~\ref{alg:alg}.

\begin{figure*}
    \centering
    \includegraphics{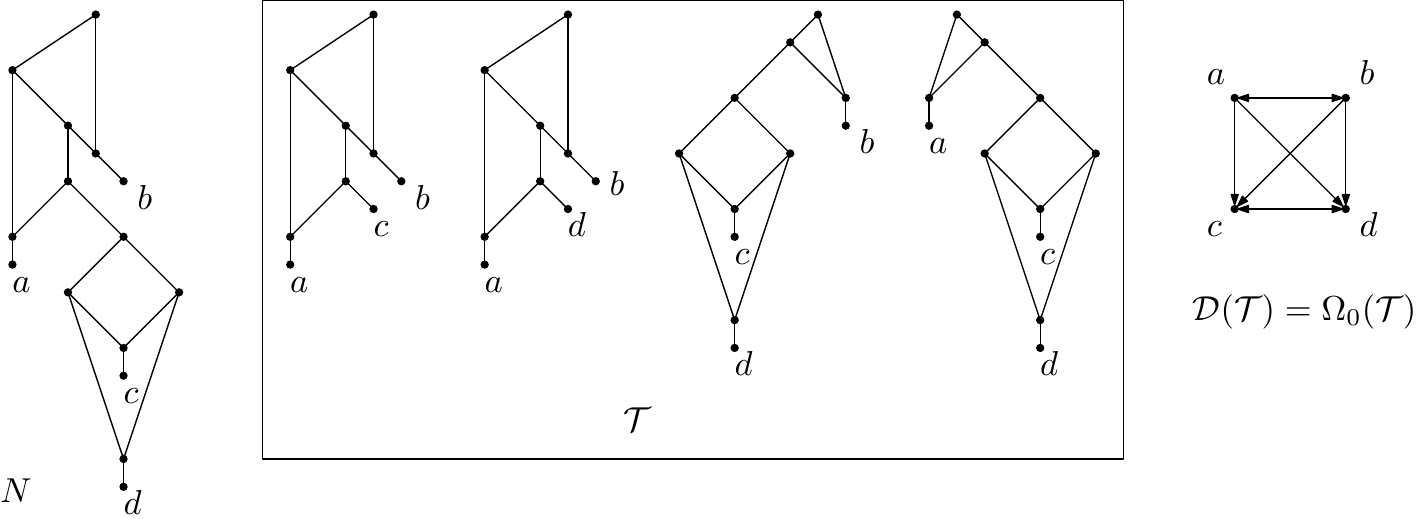}
    \caption{A level-2 network~$N$, its set of trinets~$\cT=\cT(N)$ and the digraph~$\Omega_0(\cT)=\cD(\cT)$. The set~$\{c,d\}$ is the only minimal sink set in $\Omega_0(\cT)$ and the only minimal cut-arc set in~$N$. \label{fig:cutarcs}}
\end{figure*}

\SetKwComment{Comment}{//}{}
\begin{algorithm}[h]
\KwData{Multiset~$\cT$ of \leonie{level-2} trinets and (possibly) binets on taxon set~$X$.}
\KwResult{Level-2 phylogenetic network~$N$ on~$X$.}
Find a cut-arc set~$A$ using Algorithm~\ref{alg:cut-arc-set}\;
\Comment{Find network~$N^*$ with~$A$ collapsed}
Initialize $\cT^*=\emptyset$ and let~$a^*\notin X$ be a new taxon\;
\For{$T\in\cT$ with $L(T)\setminus A\neq\emptyset$}{
\lIf{$L(T)\cap A=\emptyset$}
{Add~$T$ to~$\cT^*$}
\Else{Pick~$a\in L(T)\cap A$\;
Construct $T|(L(T)\setminus A\cup\{a\})$\;
Relabel~$a$ to~$a^*$ and add the resulting trinet or binet to~$\cT^*$\;
}
}
Construct~$N^*$ from~$\cT^*$ by recursively running this algorithm\;
\Comment{Find simple network~$N'$ on~$A$}
$\cT' \getss \{T|(L(T)\cap A) \,\mid\, T\in\cT\ ,\, |L(T)\cap A|\geq 2\}$\;
Construct a simple network~$N'$ for~$\cT'$ using Algorithm~\ref{alg:simple}\;
\Comment{Combine~$N'$ and~$N^*$}
\If{$A\neq X$}{
\Return the network constructed from~$N^*$ and~$N'$ by identifying~$a^*$ with the root of~$N'$
}\lElse{
\Return $N'$
}
\caption{Constructing level-2 networks from trinets\label{alg:alg}}
\end{algorithm}

Within our \leo{explanation} of the algorithm we will also explain why in case the underlying set of $\cT$ is $\tb(N)$ for some recoverable level-2 network~$N$, the algorithm correctly reconstructs~$N$.

\subsection{Finding a minimal cut-arc set}\label{sec:cutarcs}

We first find minimal cut-arc sets using the digraphs~$\Omega_i(\cT)$ which were introduced in~\cite{trilonet} for level-1 networks and are defined as follows. See Figure~\ref{fig:cutarcs} for an example.

\begin{definition}
\label{def:omega}
Given a multiset~$\cT$ of binets and trinets and~$i\geq 0$, $\Omega_i(\cT)$ is the digraph with vertex set~$L(\cT)$ and an arc~$(x,y)$ if at most~$i$ trinets~$T\in\cT$ with~$x,y\in L(T)$ have a minimal cut-arc set not containing~$y$.
\end{definition}

A \emph{sink set} in a digraph~$D=(V,A)$ is a set~$U\subseteq V$ such that there is no arc~$(u,v)\in A$ with~$u\in U$ and~$v\notin U$. A sink set~$U$ is \emph{minimal} if~$|U|>1$ and there is no sink set~$W$ with~$|W|>1$ and~$W\subsetneq U$. A \emph{strongly connected component} of a digraph is a maximal subgraph~$D'=(V',A')$ containing, for any~$u,v\in V'$, a directed path from~$u$ to~$v$ and from~$v$ to~$u$.

If~$N$ is a level-1 network, minimal sink sets in $\Omega_i(\cT(N))$ correspond to minimal cut-arc sets in~$N$~\cite{trilonet}. To extend this result to level-2 networks, we will use the following theorem, which is a special case of \cite[Theorem~7.3]{huber2019hierarchies}. It uses the \emph{closure digraph}~$\cD(\cT)$ of a set~$\cT$ of trinets, which was introduced in~\cite{trilonet} and is defined as follows. Its vertex set is~$X=\cup_{T\in\cT}L(T)$ and it has an arc~$(x,y)$ if, for all~$z\in X\setminus\{x,y\}$, there exists a trinet on~$\{x,y,z\}$ in~$\cT$ in which~$y$ is a descendant of $LSA(x,z)$.

\begin{theorem}\cite{huber2019hierarchies}\label{thm:sinksets} Let~$N$ be a binary level-2 network on~$X$ and~$A\subseteq X$. Then~$A$ is minimal cut-arc set of~$N$ if and only if~$A$ is a minimal sink set of the closure digraph~$\cD(\cT(N))$.
\end{theorem}


The next lemma shows that the closure digraph $\cD(\cT)$ is equal to $\Omega_0(\cT)$ if~$\cT$ is the set of trinets of some network.

\begin{lemma}\label{lem:closure}
If~$\cT=\cT(N)$ for some network~$N$ on~$X$, then~$\Omega_0(\cT)=\cD(\cT)$.
\end{lemma}
\begin{proof}
First let~$(x,y)$ be an arc of~$\Omega_i(\cT)$. Assume that~$(x,y)$ is not an arc of $\cD(\cT)$. Then there exists a $z\in X\setminus\{x,y\}$ such that $y$ is not a descendant of $LSA(x,z)$ in the trinet~$T$ on~$\{x,y,z\}$. We now claim that the arc entering $LSA(x,z)$ is a cut-arc of~$T$. If it is not, then there is some arc $(u,v)$ of~$T$ with $v\neq LSA(x,z)$ such that~$u$ is not a descendant of~$LSA(x,z)$ and~$v$ is a descendant of~$LSA(x,z)$. This arc~$(u,v)$ must lie on a path from the root to at least one of~$x,y,z$. However, it cannot be on a path from the root to~$x$ or~$z$ because each such path passes through $LSA(x,z)$. Also, it cannot be on a path from the root to~$y$ because such a path does not contain any descendants of~$LSA(x,z)$. Hence, we can conclude that~$\{x,z\}$ is a cut-arc set, which contradicts the assumption that $(x,y)$ is an arc of~$\Omega_i(\cT)$.

Now let~$(x,y)$ be an arc of~$\cD(\cT)$ and let~$z\in X\setminus\{x,y\}$. Then~$y$ is a descendant of $LSA(x,z)$ in the trinet on~$\{x,y,z\}$ in~$\cT$. Hence, $\{x,z\}$ is not a cut-arc set. Since a minimal cut-arc set contains at least two leaves, it follows that~$T$ has no minimal cut-arc set not containing~$y$. It now follows that~$(x,y)$ is an arc of~$\Omega_0(\cT)$.
\end{proof}

Since we consider trinet sets that are not necessarily exactly the trinet set of some network, we cannot always simply use the digraph $\Omega_0(\cT)=\cD(\cT)$. In particular, it may happen that $\Omega_0(\cT)$ has no arcs. We therefore use the \leo{strategy described in Algorithm~\ref{alg:cut-arc-set}.}

\begin{algorithm}[h]
\KwData{Multiset~$\cT$ of \leonie{level-2} trinets and (possibly) binets on taxon set~$X$.}
\KwResult{Set~$A\subseteq X$.}
\For{$i=0,\ldots ,|X|-2$}{
Construct $\Omega_i(\cT)$ (see Definition~\ref{def:omega})\;
\If{$\Omega_i(\cT)$ has at least one arc}{
\If{$\Omega_i(\cT)$ has a strongly connected component that is a minimal sink set}{
\Return a smallest such set\;
}
\Else{Construct the condensed digraph~$C$ of $\Omega_i(\cT)$\;
Find a vertex~$p$ of~$C$ with a minimum number of children, over all vertices with at least one child\;
\Return the set of vertices of $\Omega_i(\cT)$ corresponding to~$p$ and its children\;}
}
}
\caption{Finding a cut-arc set\label{alg:cut-arc-set}}
\end{algorithm}



From Theorem~\ref{thm:sinksets} and Lemma~\ref{lem:closure} follows that this process produces a minimal cut-arc set if the input set is equal to~$\cT(N)$ for some level-2 network~$N$. Since~$\Omega_0(\cT)$ is not affected by binets or multiple copies of trinets, the same holds when~$\cT$ is a multiset of binets and trinets with underlying set~$\tb(N)$.

\subsection{Constructing a simple network}\label{sec:simple}

Once we have found a minimal cut-arc set~$A$, we need to construct the part of the network below this cut-arc. To do this, we restrict~$\cT$ to~$\cT' = \{T|(L(T)\cap A) \,\mid\, T\in\cT,\, |L(T)\cap A|\geq 2\}$ and find a simple network for~$\cT'$.

If the underlying set of~$\cT$ is~$\tb(N)$ with~$N$ a level-2 network and~$A$ is a minimal cut-arc set of~$N$, then the underlying set of~$\cT'$ is~$\tb(N')$ with~$N'$ either a tree with two leaves or a simple network.

\subsubsection{The number of reticulations}\label{sec:retic}

Let~$p_k$ be the fraction of the trinets in~$\cT'$ that are strictly level-$k$ and let~$n=|L(\cT')|$. If~$n=2$, we construct a network equal to a binet with maximum multiplicity in~$\cT$. Otherwise, if $p_2 < \frac{n-2}{2{n \choose 3}}$, we set the number of reticulations~$k$ to~1, else we set~$k$ to~$2$.

Suppose~$\cT'$ has underlying set $\tb(N')$ with~$N'$ either a tree with two leaves or a level-2 network that is simple (note that it may also be level-1). If~$N'$ has two leaves then all binets in~$\cT'$ are equal to~$N'$ and the algorithm correctly constructs~$N'$. Now assume~$n\geq 3$. If~$N'$ is a simple level-1 network, then $p_2=0$, so the algorithm correctly sets the number of reticulations to~$1$. Finally, suppose~$N'$ is a simple strictly level-2 network. Then $|\cT'| = {n\choose 3}$. Moreover, at least $n-2$ of the trinets in~$\cT'$ are strictly level-2, since any crucial trinet is strictly level-2. Hence, we have $p_2 \geq \frac{n-2}{{n \choose 3}} \geq \frac{n-2}{2{n \choose 3}}$ and the algorithm correctly sets the number of reticulations to~$2$.

\subsubsection{Leaves on reticulation sides}

Let~$k$ be the number of reticulations determined in the previous subsection. Let~$G$ be a generator that is the underlying generator of the maximum number of strictly level-$k$ trinets in~$\cT$. Let~$\cT_G$ be the set of trinets in~$\cT$ that have underlying generator~$G$. 

For each~$x\in L(\cT_G)$ and for each set of symmetric reticulation sides~$C$ of~$G$, let $p_{x,C}$ denote the fraction of trinets in~$\cT_G$ that have leaf~$x$ on a side in~$C$. We proceed greedily as follows. Pick~$x,C$ maximizing~$p_{x,C}$ over all leaves~$x$ that have not been assigned to a side yet and over all~$C$ containing at least one side that has not been assigned a leaf yet. Assign~$x$ to an arbitrary side in~$C$. Repeat until all reticulation sides have been assigned a leaf. Attach each leaf assigned to a reticulation side to this side. 

Let~$\cT_G^r$ be the set of trinets in~$\cT$ that have underlying generator~$G$ and that have an automorphism such that each reticulation side of~$G$ contains its assigned leaf. From now on, we assume that each reticulation side of the generator of each trinet in~$\cT_G^r$ contains its assigned leaf.

Suppose the underlying set of~$\cT$ is~$\tb(N)$ for some simple level-2, strictly level-$k$, network~$N$. Then all strictly level-$k$ trinets have the same underlying generator as~$N$. Moreover, for each set~$C$ of symmetric reticulation sides, $p_{x,C}=1$ for all leaves~$x$ that are on a side in~$C$ in~$N$ and $p_{x,C}=0$ otherwise. Hence, the algorithm correctly assigns leaves to sets of symmetric reticulation sides. It can assign leaves to an arbitrary side within this set since level-2 generators have at most one set of symmetric reticulation sides (see Figure~\ref{fig:gen}), and those are symmetric.

\subsubsection{Leaves per set of symmetric arc sides}
For each leaf~$x\in L(\cT^r_G)$ that has not been assigned to a reticulation side, assign~$x$ to a set of symmetric arc sides~$C$ of~$G$, maximizing the fraction of trinets in~$\cT^r_G$ that have leaf~$x$ on a side in~$C$.


Suppose the underlying set of~$\cT$ is~$\tb(N)$ for some simple level-2 network~$N$. Then it can be argued as in the previous subsection that the algorithm assigns each leaf to the set of symmetric arc sides corresponding to its location in~$N$.

\subsubsection{Leaves per arc side}

Consider a set of symmetric arc sides~$C$ and the set of leaves~$X_C$ assigned to~$C$. For~$x,y\in X_C$, let~$\cT_{xy}$ denote the set of simple trinets in~$\cT$ containing both~$x$ and~$y$, and let~$q_{xy}$ denote the fraction of trinets in~$\cT_{xy}$
in which~$x$ and~$y$ are on the same side of the underlying generator, with~$q_{xx}=1$. We define the following score for~$x\neq y$:

\begin{equation*}
r_{xy} = 3\sum_{z\in X_C} \min\{q_{xz},q_{yz}\} - \sum_{z\in X_C} q_{xz} - \sum_{z\in X_C} q_{yz}.
\end{equation*}

The main idea of this score function is that, assuming the trinets come from some level-2 network, $r_{xy}\geq 0$ if and only if~$x$ and~$y$ are on the same side.

The algorithm proceeds as follows. Create a partition~$\cP_C$ of~$X_C$, initially consisting of only singletons. While $|\cP_C|>|C|$ or there exist $x\neq y$ with~$r_{xy}>0$, pick a pair~$X,Y\in \cP_C$ maximizing
\begin{equation}
r_{XY}=\frac{1}{|X||Y|}\sum_{x\in X, y\in Y}r_{xy}.\label{eq:r}
\end{equation}
Merge sets~$X$ and~$Y$ in~$\cP_C$.

Finally, assign, injectively at random, the parts of~$\cP_C$ to the sides in~$C$.

Suppose the underlying set of~$\cT$ is~$\tb(N)$ for some simple level-2 network~$N$. The only level-2 generators with symmetric arc sides (see Figure~\ref{fig:gen}) are $1$ and~$2d$ with $C=\{L,R\}$ and $2c$ with~$C=\{L_i,R_i\}$, $i\in\{1,2,3\}$. If~$x,y$ are on the same side then~$q_{xy}=1$ and otherwise we have~$q_{xy}=0$. We can now see that if~$x,y$ are on the same side then $r_{xy}$ is equal to the number of leaves on that side (since each of the three sums is equal to the number of leaves on that side) which is at least~$2$. If, on the other hand,~$x,y$ are on different sides, then~$r_{xy}\leq -2$ (since the first sum is~$0$ and the other two sums are at least~$1$). Hence, the algorithm correctly splits the leaves in~$X_C$ into two sets corresponding to the leaves on side~$L_i$ and~$R_i$ (or~$L$ and~$R$). For generators~$1$ and~$2d$ it does not matter which set is assigned to which side, by symmetry. For generator~$2c$, this does matter. It is done randomly here and corrected if necessary in the next subsection.

\subsubsection{Side alignment}
The following is only necessary when the underlying generator~$G$ is generator~$2c$, see Figure~\ref{fig:gen}, since it contains more than one set of symmetric arc sides. Call its sets of symmetric arc sides $C_1=\{L_1,R_1\}$, $C_2=\{L_2,R_2\}$ and~$C_3=\{L_3,R_3\}$. We have to consider swapping sides $L_2,R_2$ and/or $L_3,R_3$ (i.e., assign the leaves assigned to~$L_2$ to~$R_2$ and vice versa and/or assign the leaves assigned to~$L_3$ to~$R_3$ and vice versa). From the four possibilities, we choose the one maximizing the following score:
\begin{equation}\label{eq:usum}
u_{L_1,L_2} + u_{L_1,L_3} + u_{L_2,L_3} + u_{R_1,R_2} + u_{R_1,R_3} + u_{R_2,R_3}
\end{equation}
with
\begin{equation}\label{eq:u}
u_{S,T} = \sum_{x\in X_S\\ y\in X_{T}} q_{xy} - |X_S||X_T|.
\end{equation}

Suppose the underlying set of~$\cT$ is~$\tb(N)$ for some simple level-2 network~$N$ with underlying generator~$2c$. Then we have that~$q_{xy}=1$ if~$x\in L_i,y\in L_j$ or~$x\in R_i,y\in R_j$ and $q_{xy}=0$ if~$x\in L_i,y\in R_j$ or vice versa. Hence,~$u_{L_iL_j}=u_{R_iR_j}=0$ and $u_{L_iR_j},u_{R_iL_j}<0$. Therefore, choosing the assignment maximizing~\eqref{eq:usum}, out of all possible assignments, chooses the assignment corresponding to~$N$.

\subsubsection{Ordering the leaves on the arc sides}

Consider a side~$S$ and the set of leaves~$X_S$ assigned to side~$S$. Let~$\cT^s_{xy}$ denote the set of simple trinets in~$\cT$ containing both~$x$ and~$y$ and both on the same side. Let~$a_{xy}$ denote the fraction of trinets in~$\cT^s_{xy}$ in which the parent of~$x$ is an ancestor of~$y$. Let~$\pi$ be an ordered list of leaves, which is initially empty. Find a leaf~$x\in X_S\setminus\pi$ maximizing
\begin{equation}\label{eq:asum}
\sum_{y\in X_S\setminus\pi} a_{xy}-a_{yx}.
\end{equation}
Append leaf~$x$ to~$\pi$ and continue until~$\pi$ is a permutation of~$X_S$. The permutation~$\pi$ then describes the ordering of the leaves on side~$S$. Attach the list of leaves~$\pi$ to side~$S$.

Suppose the underlying set of~$\cT$ is~$\tb(N)$ for some simple level-2 network~$N$. For two leaves~$x,y$ on the same arc side~$S$ of~$N$, we have that~$a_{xy}=1$ if the parent of~$x$ is an ancestor of~$y$ and~$a_{xy}=0$ otherwise. Hence, \eqref{eq:asum} is equal to the number of leaves that have not been added to the permutation~$\pi$ yet and are below~$x$ on side~$S$, minus the number of leaves that have not been added to the permutation~$\pi$ yet and are above~$x$ on side~$S$. Therefore, the algorithm constructs the ordering~$\pi$ of leaves on side~$S$ in~$N$.

\leo{The pseudo code for constructing a simple network is in Algorithm~\ref{alg:simple}.}

\begin{algorithm}
\KwData{Multiset~$\cT'$ of \leonie{level-2} trinets and (possibly) binets on taxon set~$X$.}
\KwResult{Simple level-2 network~$N'$ on~$X$.}
\Comment{Determine the level~$k$}
Let $n\getss |L(\cT')|$ and~$p_2\getss$ the fraction of trinets in~$\cT'$ that are level-$2$\;
\If{$n=2$}{
\Return an arbitrary network with maximum multiplicity in~$\cT'$
}
\If{$p_2<\frac{n-2}{2 {{n}\choose {3}}}$}{
$k\getss 1$
}\Else{
$k\getss 2$
}
\Comment{Determine the generator}
$G\getss$ the underlying generator of the maximum number of level-$k$ trinets in~$\cT'$\;
$N'\getss G$\;
\Comment{Assign leaves to reticulation sides}
$\cT_G\getss$ the set of trinets in~$\cT'$ that have underlying generator~$G$\;
\While{there is a reticulation side of~$G$ that has not been assigned a leaf}{
Find~$x\in L(\cT_G)$ that has not been assigned to a side and a set of symmetric reticulation sides~$C$ that have not all been assigned a leaf, \leo{maximizing} the fraction of trinets in~$\cT_G$ that have leaf~$x$ on a side in~$C$\;
Assign~$x$ to an arbitrary side in~$C$ and attach~$x$ to this side in~$N'$\;
}
$\cT_G^r\getss$ the set of trinets in~$\cT'$ that have underlying generator~$G$ and that have an automorphism such that each reticulation side of~$G$ contains its assigned leaf\; 
Relabel the sides of the generators of the trinets in~$\cT_G^r$ such that each reticulation side contains its assigned leaf\;
\Comment{Assign leaves to sets of symmetric arc sides}
\For{each leaf~$x$ that has not been assigned to a reticulation side}{
Assign~$x$ to a set of symmetric arc sides~$C$ maximizing the fraction of trinets in~$\cT^r_G$ that have leaf~$x$ on a side in~$C$\;
}
Continued in Algorithm~\ref{alg:continued}
\caption{Constructing a simple level-2 network\label{alg:simple}}
\end{algorithm}
\begin{algorithm}
\Comment{Assign leaves to arc sides}
\For{each set of symmetric arc sides~$C$}{
$\cP_C\getss$ partition of~$X_C$ containing only singletons\;
$q_{xy}\getss$ the fraction of simple trinets containing~$x,y$ in which~$x,y$ are on the same side\;
$\displaystyle r_{xy}\getss 3\sum_{z\in X_C} \min\{q_{xz},q_{yz}\} - \sum_{z\in X_C} q_{xz} - \sum_{z\in X_C} q_{yz}$\;
$\displaystyle r_{XY}\getss \frac{1}{|X||Y|}\sum_{x\in X,y\in Y}r_{xy}$\;
\While{there exist $X,Y\in\cP_C$ with~$r_{XY}>0$, or $|\cP_C|>|C|$}{
\label{line:r} Find a pair~$X,Y\in \cP$ maximizing $r_{XY}$\; 
Merge sets~$X$ and~$Y$ in~$\cP_C$ and update~$r_{XY}$\;
}
\While{there is a leaf in~$X_C$ that has not been assigned to a side}{
Pick a set $Z\in\cP_C$ containing a leaf that has not been assigned to a side\;
Pick a side~$S\in C$ that has not been assigned any leaves\;
Assign the leaves from~$Z$ to side~$S$\;
}
}
\Comment{Align sides}
\If{$G$ is generator $2c$ from Figure~\ref{fig:gen}}{
Find bijections $f:\{L_2,R_2\}\to \{L_2,R_2\}$ and $g:\{L_3,R_3\}\to \{L_3,R_3\}$ maximizing 
$u_{L_1,f(L_2)} + u_{L_1,g(L_3)} + u_{f(L_2),g(L_3)} + u_{R_1,f(R_2)} + u_{R_1,g(R_3)} + u_{f(R_2),g(R_3)}$\;
with $u_{S,T} = \displaystyle\sum_{x\in X_S, y\in X_{T}} q_{xy} - |X_S||X_T|$\;
Assign the leaves assigned to~$L_2,R_2,L_3,R_3$ to $f(L_2),f(R_2),g(L_3),g(R_3)$, respectively\;
}
\Comment{Order leaves on arc sides}
\For{each arc side~$S$ with set~$X_S$ of assigned leaves}{
$\cT^s_{xy}\getss$ the set of simple trinets in~$\cT'$ containing~$x$ and~$y$ on the same side\;
$a_{xy}\getss$ the fraction of trinets in~$\cT^s_{xy}$ in which the parent of~$x$ is an ancestor of~$y$\;
$\pi\getss ()$\;
\While{$\pi$ is not a permutation of~$X_S$}{
Find a leaf~$x\in X_S\setminus\pi$ maximizing $\sum_{y\in X_S\setminus\pi} a_{xy}-a_{yx}$ 
\leo{and append}~$x$ to~$\pi$\;
}
Attach the list of leaves $\pi$ to side~$S$ in~$N'$\;
}
\Return $N'$\;
\caption{Continuation of Algorithm~\ref{alg:simple}\label{alg:continued}}
\end{algorithm}



\subsection{Theoretical result}


The following theorem shows that the algorithm is guaranteed to reconstruct a level-2 network from its set of trinets.


\begin{theorem}\label{thm:induced}
If~$N$ is a recoverable, binary level-2 network on~$X$ with~$|X|\geq 3$, then Algorithm~\ref{alg:alg} will output~$N$ when applied to input~$\cT=\cT(N)$.
\end{theorem}
\begin{proof}
The proof is by induction on the number of vertices of~$N$.

The base case is that~$N$ is a tree with~$3$ leaves and~$5$ vertices. Say that~$X=\{x,y,z\}$ and that~$\{x,y\}$ is the minimal cut-arc set. In this case, the algorithm will generate $A=\{x,y\}$ (see Section~\ref{sec:cutarcs}). The set~$\cT'$ contains only the tree on~$\{x,y\}$ and hence this is constructed as~$N'$ (see Section~\ref{sec:retic}). The set~$\cT^*$ contains only the tree on~$\{a^*,z\}$ and hence~$N^*$ is this tree. Combining~$N'$ and~$N^*$ then gives~$N$ (see Section~\ref{sec:combine}).

Now suppose~$N$ has at least~$6$ vertices. Then the algorithm finds a minimal cut-arc set~$A$ of~$N$ by Section~\ref{sec:cutarcs}. Let~$(u,v)$ be the corresponding cut-arc of~$N$ and let~$N'$ be the subnetwork of~$N$ rooted at~$v$. Let~$N^*$ be the network obtained from~$N$ by deleting all vertices of~$N'$ except for~$v$ and labelling~$v$ by~$a^*$. Then the underlying set of~$\cT'$ is~$\tb(N')$ and the underlying set of~$\cT^*$ is~$\tb(N^*)$. We have argued in Section~\ref{sec:simple} that the algorithm constructs~$N'$ (which is either a tree with two leaves or a simple network) from~$\cT'$. If~$A=X$ then~$N=N'$ since~$N$ is recoverable and we are done. Otherwise,~$N^*$ contains fewer vertices than~$N$. If~$N^*$ has at least three leaves, the algorithm constructs~$N^*$ from~$\cT^*$ by induction. If~$N^*$ has two leaves, then~$\cT^*$ only contains~$N^*$ and hence the algorithm constructs~$N^*$ (see Section~\ref{sec:retic}). In both cases, combining~$N'$ and~$N^*$ gives~$N$ (see Section~\ref{sec:combine}).
\end{proof}

It remains to analyze the running time of the algorithm.
Algorithm~\ref{alg:cut-arc-set} can be implemented efficiently to run in~$O(|\cT|+|X|^2)$ time (similarly to~\cite{trilonet} for level-1). The main idea here is to first compute $\phi(x,y)$, the number of trinets containing~$x$ and~$y$ that have a minimal cut-arc set not containing~$y$. This can be done in~$O(|\cT|+|X|^2)$ time since we need to loop through the set of trinets only once and update the values~$\phi(x,y)$ affected by this trinet~$T$, i.e., with~$x,y\in L(T)$. Finding a minimal cut-arc set in a trinet can be done in constant time as the size of each trinet is bounded by a constant \leonie{(as any trinet that is not recoverable can be ignored)}. After that, the digraph~$\Omega_i$ can be constructed in~$O(|X|^2)$ time, and this only needs to be done for the smallest~$i$ for which~$\phi(x,y)\leq i$ for at least one pair~$x,y$. The condensed digraph can be found with Tarjan's algorithm for computing strongly connected components in~$O(|X|^2)$ time. Since the number of generators, and the number of sides of each generator, is bounded by a constant, the bottleneck of Algorithm~\ref{alg:simple} is Line~\ref{line:r}. The values~$q_{xy}$ can be computed in~$O(|\cT|+|X|^2)$ time and the values~$r_{xy}$ in~$O(|X|^3)$ time. The values~$r_{XY}$ can be computed in~$O(|X|^2)$ time by looping through all~$x,y$ and updating the values of~$r_{XY}$ with~$x\in X$ and~$y\in Y$. This last step has to be repeated~$O(|X|)$ times. So Algorithm~$\ref{alg:simple}$ takes $O(|\cT|+|X|^3)$ time. Computing~$\cT'$ and~$\cT^*$ can be done in~$O(|\cT|+|X|)$ time since the size of the trinets is bounded by a constant. All of this has to be repeated~$O(|X|)$ times. Hence, the algorithm runs in time $O(|\cT||X|+|X|^4)$.

\section{Discussion}

We have presented an algorithm that, for an input set $\mathcal T$ of trinets (and possibly binets) with leaf-set $X$, outputs a level-2 network on $X$ with run time $O(|\cT||X|+|X|^4)$ and that is guaranteed to reconstruct a level-2 network from its set of trinets. Note that a variant of this algorithm is presented in ~\cite{sjors}. It should also be noted that our level-2 algorithm cannot be used to decide whether or not an arbitrary set of trinets is contained in some level-2 network or not in polynomial time. Indeed, if an arbitrary set of level-1 trinets is input into the algorithm, then it will output a level-1 network.
But it is known that deciding whether or not an arbitrary set of level-1 trinets is contained in a level-1 network is NP-complete \cite{huber2017reconstructing}.
In addition, in light of these observations concerning level-1 trinets, our algorithm can be used build level-1 networks for more general inputs that the level-1 TriLoNet algorithm described in \cite{trilonet}, since TriLoNet's input is restricted to collections in which there is a trinet on every 3-subset of the leaf-set.

In terms of potential applications of our level-2 algorithm, in \cite{trilonet} a method is presented to derive collections of level-1 trinets from molecular sequence data; it would be interesting to see if this approach could be extended (or a new approach developed) to derive level-2 trinets as well. We expect that this could be quite complicated as level-2 trinets (and even level-1 trinets) can be quite complex, and so it may be necessary to restrict the level-1/level-2 building blocks to some subset of the list of potential 3-leaved networks.  

In another direction, in this paper we have shown that level-3 networks are not necessarily encoded by their trinets. However, 
Figure~\ref{fig:counter_example} is essentially the only case in which a level-3 network is not encoded~\cite{leonie}, and so it would be interesting to investigate if there is a polynomial-time algorithm for constructing level-3 networks from trinets modulo this symmetry. Alternatively, it can be shown that the collection of 4-leaved networks (or quarnets) contained in a level-3 network encode the network~\cite{leonie}, and so new algorithms could be potentially developed to build level-3 networks from quarnets. In this vein, an interesting open question is whether or not a level-$k$ network is always encoded by its $(k+1)$-nets. Some partial results concerning this question are presented in~\cite{frank}.

\bibliographystyle{plain}
\bibliography{bib}
\end{document}